\documentclass[acmsmall]{acmart}

\AtBeginDocument{%
  \providecommand\BibTeX{{%
    \normalfont B\kern-0.5em{\scshape i\kern-0.25em b}\kern-0.8em\TeX}}}

\setcopyright{acmcopyright}
\copyrightyear{2018}
\acmYear{2018}
\acmDOI{10.1145/1122445.1122456}

\acmJournal{JACM}
\acmVolume{37}
\acmNumber{4}
\acmArticle{111}
\acmMonth{8}

\usepackage{xcolor}
\usepackage{amsfonts,mathrsfs}
\usepackage{amsmath}
\usepackage{bbm}
\usepackage{algpseudocode,algorithm}
\usepackage{mathtools}
\usepackage{verbatim}
\usepackage{acronym}
\usepackage{graphicx}
\usepackage{subfigure}

\newtheorem{remark}{Remark}

\def\e{\epsilon}
\def\d{\delta}

\DeclarePairedDelimiter\floor{\lfloor}{\rfloor}

\begin{document}

\title{{Optimal Online Algorithms for File-Bundle Caching and Generalization to Distributed Caching}}


\author{Tiancheng Qin}
\affiliation{%
 \institution{Department of Industrial and Systems Enginering\\ University of Illinois at Urbana-Champaign}
 \streetaddress{104 S. Mathews Ave.}
 \city{Urbana}
 \state{IL}
 \country{USA}}
\email{tq6@illinois.edu}

\author{S. Rasoul Etesami}\thanks{This work is supported by the National Science Foundation CAREER Award under Grant No. EPCN-1944403.}
\affiliation{%
	\institution{Department of Industrial and Systems Engineering, Coordinated Science Lab\\ University of Illinois at Urbana-Champaign}
	\streetaddress{104 S. Mathews Ave.}
	\city{Urbana}
	\state{IL}
	\country{USA}}
\email{etesami1@illinois.edu}

\renewcommand{\shortauthors}{Qin and Etesami}

\begin{abstract}
We consider a generalization of the standard cache problem called \emph{file-bundle} caching, where different queries (tasks), each containing $l\ge 1$ files, sequentially arrive. An online algorithm that does not know the sequence of queries ahead of time must adaptively decide on what files to keep in the cache to incur the minimum number of cache misses. Here a cache miss refers to the case where at least one file in a query is missing among the cache files. In the special case where $l=1$, this problem reduces to the standard cache problem. We first analyze the performance of the classic \emph{least recently used} (LRU) algorithm in this setting and show that LRU is a near-optimal online deterministic algorithm for file-bundle caching with regard to competitive ratio. We then extend our results to a generalized $(h,k)$-paging problem in this file-bundle setting, where the performance of the online algorithm with a cache size $k$ is compared to an optimal offline benchmark of a smaller cache size $h<k$. In this latter case, we provide a randomized $O(l \ln \frac{k}{k-h})$-competitive algorithm for our generalized $(h,k)$-paging problem, which can be viewed as an extension of the classic \emph{marking algorithm}. We complete this result by providing a matching lower bound for the competitive ratio, indicating that the performance of this modified marking algorithm is within a factor of two of any randomized online algorithm. Finally, we look at the distributed version of the file-bundle caching problem where there are $m\ge 1$ identical caches in the system. In this case we show that for $m=l+1$ caches, there is a deterministic distributed caching algorithm which is $(l^2+l)$-competitive and a randomized distributed caching algorithm which is $O(l\ln(2l+1))$-competitive when $l\ge 2$. We also provide a general framework to devise other efficient algorithms for the distributed file-bundle caching problem and evaluate the performance of our results through simulations. 
\end{abstract}



\maketitle

\section{Introduction}
Caching has long been considered as an effective and essential method to improve the performance of storage systems \cite{satyanarayanan1990survey}. In this regard, a basic and well-studied model for the cache problem is the two-level memory model, where the storage system has a large but slow main memory and a small but fast cache. In this model, the main memory contains a large set of $N$ files while the smaller cache can store at most $k\ll N$ files. Whenever a client requests some file, if it is available in the cache, then it can be sent to him immediately, and it is called a cache \emph{hit}; otherwise, it is called a cache \emph{miss} and incurs a cost (e.g., delay or energy cost).

While in the literature of cache problem files are usually considered to be requested one by one, and a cache miss is defined as missing a single requested file, nowadays it frequently arises in scientific data management applications such as data-grids, when conducting analysis for data-intensive scientific experiments such as High Energy Physics(HEP) experiments, every task (an analysis process) will request a bundle of files, each containing different attributes of some object, and it can be serviced only if all its requested files are in the cache at the same time\cite{otoo2004optimal}\cite{shoshani2000coordinating}. In this setting, a cache miss is thus defined to be missing any of the files in the file-bundle requested by the task. Such an extension of the standard cache problem is often referred to as the \emph{file-bundle caching} problem \cite{otoo2004optimal,otoo2005impact}. In the distributed setting, we further assume that a query may be served by a distributed storage system instead of a centralized storage system. In fact, the offline version of the distributed file-bundle caching problem has been studied in \cite{etesami2016design}, where a fundamental trade-off between cache design cost and cache access cost to respond to every query has been established. 

In this paper, we propose a novel approach to the \emph{online} version of the file-bundle caching problem in which every query consists of $l\ge 1$ files, and a cache miss occurs if at least one of the files in the query is not in the cache. The goal is then to devise online algorithms to minimize the total cache misses over an arbitrary sequence of query arrivals; that is, we have to decide the cache's content without any information on the future queries. The file-bundle caching problem is indeed a generalization of the standard cache problem since the latter can be viewed as a special case of the former when $l=1$. Subsequently, we also look at the distributed version of the file-bundle caching problem when there are $m>1$ caches in the system and design online algorithms which are competitive with respect to an optimal offline benchmark, i.e., the one which knows the entire sequence of queries a priori.

\subsection{Related Work}
The cache problem, also known as the paging problem, is one of the earliest problems studied in the field of competitive analysis and has a rich literature \cite{irani1996competitive}. For deterministic algorithms, it is known that the least recently used (LRU) algorithm that at each cache miss inserts the new file into the cache by evicting the least recently used file is $k$-competitive, and that is the best one can hope for. Here, $k$ refers to the cache size of the online algorithm, which is the same as that in the optimal offline algorithm \cite{sleator1985amortized}. When randomization is allowed, \cite{fiat1991competitive} introduced the marking algorithm and showed that it is $2H_k$-competitive.\footnote{Here, $H_k:=\sum_{i=1}^{k}\frac{1}{i}\approx\ln k$ is the $k$-th harmonic number.} In particular, it was shown that no randomized algorithm could achieve a competitive ratio better than $H_k$. Subsequently, $H_k$-competitive algorithms which exactly match this lower bound were given in \cite{mcgeoch1991strongly} and \cite{achlioptas2000competitive}.

For the ($h,k$)-paging problem, where the online algorithm with cache size $k$ is compared to the weaker offline optimal algorithm with cache size $h\le k$, \cite{sleator1985amortized} showed that no deterministic algorithm can achieve a competitive ratio better than $\frac{k}{k-h+1}$, and that LRU is exactly $\frac{k}{k-h+1}$-competitive. When randomization is allowed, \cite{young1991line} showed that the marking algorithm is $2\log \frac{k}{k-h+1}$-competitive (omitting lower order terms from the competitive ratio), and is roughly within a factor of two of the optimal. However, all of the above results assume that files are always requested one by one and do not capture scenarios where queries come as batches and must be responded by a cache containing \emph{all} the batch files. It is worth noting that the cache problem can itself be viewed as a special case of a more general $k$-server problem, which has been extensively studied in the past literature \cite{koutsoupias2009k}. However, when queries consist $l\ge 1$ files and the cache misses are defined as in the file-bundle setting, then a simple reduction of the file-bundle caching problem to the $k$-server problem becomes nontrivial. In particular, the file-bundle caching problem provides a broader view and new insights for the classic cache problem, which is of its own interest.

Apart from the cache problem, there is a rich literature on distributed storage and distributed caching, where the goals are to reduce access latencies, lower network traffic, and alleviate loads on a server. In this regard, different models \cite{awerbuch1998distributed,baev2008approximation,borst2010distributed} have been proposed and various algorithmic techniques to achieve these objectives have been analyzed \cite{maddah2014fundamental,maddah2015decentralized,cao1997cost,podlipnig2003survey}. Nonetheless, as in the case of caching and ($h,k$)-paging problems, these works do not address the situations where queries consist of a set of files and must be responded by the cache upon their arrival. Therefore, assuming that every query consists of $l\ge 1$ files, in this paper, we develop online distributed caching algorithms that are particularly suitable for the scenarios where completion of a query critically depends on the completion of each of its (sub) files. 

Perhaps, one of the first studies under the file-bundle setting is \cite{otoo2004optimal} where a near-optimal offline algorithm together with an online algorithm that iteratively solves the offline problem upon the arrival of new queries are proposed. However, the online algorithm in \cite{otoo2004optimal} requires a reallocation of the entire cache whenever a new query arrives, and no performance analysis is conducted except for simulations, where the algorithm's byte miss ratio is measured and compared with other algorithms. Moreover, \cite{little2009object} considers a different version of the online file-bundle caching problem and conducts competitive ratio analysis. However, their focus is mainly on determining the trade-off between bypassing or updating the cache, thus making their setting fundamentally different from ours. In particular, the definition of the cost function in \cite{little2009object} is in terms of the cost of missing single files, which is different from our query-wise definition of cost functions.

\subsection{Contributions and Organization}
In Section \ref{sec:model}, we provide a formal definition of the problem. In Section \ref{sec:LRU}, we introduce the LRU algorithm and analyze its performance in the case of the file-bundle caching problem. More specifically, we show that no deterministic online algorithm for the file-bundle caching problem can achieve a competitive ratio better than $k-l+1$, where $k$ is the cache size and $l$ is the length of each query, and that LRU is $k$-competitive. In particular, we adopt a new notion of ``\emph{phase}" which will be of crucial importance for the rest of our analysis. In Section \ref{sec:h,k}, we introduce a generalized ($h,k$)-paging problem and a modified marking algorithm under the file-bundle setting and show that 
\begin{enumerate}
	\item the competitive ratio of the modified marking algorithm in the generalized $(h,k)$-paging problem is at most $2l (\ln \frac{k}{k-h}-\ln \ln \frac{k}{k-h}+\frac{1}{2})$, whenever $\frac{k}{k-h}\ge e$, and it is at most $2l$ otherwise.
	\item When $\frac{k}{k-h}\ge e$, for any randomized caching algorithm there exists a sequence of queries such that the expected number of cache misses by the algorithm on that instance is at least $l(\ln \frac{k+1}{k-h+l+1}-\ln\ln\frac{k+1}{k-h+l+1}-2)$ times more than that of the optimal offline algorithm. 
\end{enumerate}

In Section \ref{sec:distributed}, we get into the distributed file-bundle caching problem and provide two distributed caching algorithms for the specific case when there are $m=l+1$ caches in the system. The first one is a deterministic $(l^2+l)$-competitive algorithm while the other is a randomized $2l(\ln(2l+1)-\ln\ln l+\frac{1}{2})$-competitive algorithm given $l\ge 2$. Consequently, we provide a general framework to design distributed caching algorithms for the file-bundle caching problem by introducing a so-called \emph{virtual} cache and reducing the distributed caching problem to the generalized $(h,k)$-paging problem. In Section \ref{sec:simulation}, we evaluate the performance of our devised file-bundle algorithms using various simulations. The results here are compared with the \emph{farthest in future} (FF) algorithm, which serves as a proxy of the optimal offline algorithm, and the \emph{random eviction} algorithm. We conclude the paper with some final remarks and further directions of research in Section \ref{sec:conclusion}.

\section{Problem Formulation}\label{sec:model}

In this section, we give a formal definition of the file-bundle caching problem under both single cache and distributed (multiple) cache settings. 

\subsection{Single-cache file-bundle caching problem}
Consider a data storage system with a total set of $\mathcal{O}=\{1,2,\ldots,N\}$ pages (files), each of unit size. Moreover, let us assume that there is a single cache of size $k$ whose content at each discrete time instances $t=1,2,\ldots,T$ is denoted by $C_t\subset \mathcal{O}$. Note that as the cache's size is fixed, we have $|C_t|=k, \forall t$. At each time $t$, a new query $Q_t\subset\mathcal{O}$ requesting a subset of pages arrives in the system and must be immediately responded by the cache $C_t$. We allow replications of pages in the queries, and as a consequence, without loss of generality, we may assume the sizes of all the queries are the same and equals to $|Q_t|=l, \forall t$. \footnote{In the case where different queries have different lengths, one can always add dummy pages to the smaller queries to make all the queries have the same maximum size of $l$.} Upon the arrival of a new query $Q_t$ at time $t$, if all the pages in $Q_t$ exist in the current content of the cache, i.e., $Q_t\subseteq C_t$, then there will be no cache miss, and we incur no cost. Otherwise, if $Q_t\not\subseteq C_t$, then a cache miss occurs, and we incur a unit cost. Therefore, we define the cost at each time instance $t$ to be $f(C_t,Q_t):=\mathbbm{1}_{\{Q_t\not\subseteq C_t\}}$, which is an indication of a cache miss at time $t$ in the file-bundle caching problem. For any caching algorithm (including the optimal offline algorithm), if a cache miss happens, it must update the content of the cache from $C_{t}$ to $C_{t+1}$ by inserting the query $Q_t$ into the cache and evicting some of the current pages. Note that such an update must take place before the arrival of the new query $Q_{t+1}$ at time $t+1$. Apart from this, no change to the cache is allowed by the algorithm. In other words, the only decisions that the caching algorithm makes are which pages to evict every time a cache miss happens.  The difference between online and offline algorithms is that online algorithms must decide on the update without knowledge of upcoming queries $Q_{t+1},Q_{t+2},\ldots$, while offline algorithms always know the entire sequence of queries, including future queries.

For a given sequence of queries $\sigma=\{Q_t\}_{t=1}^{T}$ and a given caching algorithm $\mathrm{ALG}$, we let the cost of the algorithm on the requested sequence $\sigma$ be $\mathrm{ALG}(\sigma):=\sum_{t=1}^{T}f(C_t,Q_t)$. Denoting the cost of the optimal offline algorithm on the sequence of queries $\sigma$ by $\mathrm{OPT}(\sigma)$, we say $\mathrm{ALG}$ has a competitive ratio $r$ if for \emph{every} sequence $\sigma$, we have,
\begin{align}\nonumber
\mathrm{ALG}(\sigma) \le r \cdot \mathrm{OPT}(\sigma) + c,
\end{align}
where $c$ is a constant which does not depend on $\sigma$. Therefore, in the file-bundle caching problem, our goal is to find an online algorithm that has a small competitive ratio. 

\subsection{Distributed file-bundle caching problem}
Consider a data-grid with multiple geographically distributed computing clusters; each keeps a replication of the whole data set and also keeps its own cache. Once a task is submitted to the data-grid, it is matched to a suitable computing cluster for execution. If the file-bundle requested by the task is contained in the cache of one of the computing clusters, then it can be assigned to that computing cluster and executed immediately; otherwise, it must wait for some files to be fetched from the hard drive before execution. A typical example of such data-grids is GridPP (Grid for Particle Physics)\cite{britton2009gridpp}. Motivated by this observation, we propose the following distributed file-bundle caching problem.

In the distributed file-bundle caching problem, everything remains as before except that instead of having one single cache, we now have $m\ge 1$ caches of identical size $k$. In particular, our goal is to dynamically update the caches such that every query can be responded by at least one of them. More precisely, we consider a set of $m$ caches each of size $k$. The collection of all cache contents at time $t$ is given by $C_t = \{C_t^i\}_{i=1}^m$, where $C_t^i\subset\mathcal{O}$ and $|C_t^i|= k, \forall i$. Upon the arrival of a query $Q_t$ at time $t$, we say a cache hit happens if there exists at least one cache among the set of $m$ caches that can answer the entire query; otherwise, we call it a cache miss. Therefore, in this case, the cost at time $t$ is given by
\begin{align}\nonumber
f(C_t,Q_t):= \prod_{i=1}^m \mathbbm{1}_{\{Q_t\not\subseteq C_t^i\}},
\end{align}
and the overall cost of an online algorithm $\mathrm{ALG}$ over a sequence of queries $\sigma=\{Q_t\}_{t=1}^{\mathrm{T}}$ equals to $\mathrm{ALG}(\sigma)=\sum_{t=1}^{T}f(C_t,Q_t)$.  However, as before, we compare the online algorithm performance with respect to the same optimal offline benchmark as in the single-cache setting, i.e., the one with a single cache of size $k$. Again, in the distributed setting, we are interested in finding an online algorithm with a small competitive ratio.

\begin{remark}\label{def3}
	While the above definitions of the competitive ratios are given for deterministic algorithms, however, they can be extended naturally to randomized algorithms. In that case, we say that a randomized algorithm $\mathrm{ALG}$ which updates the cache contents randomly has a competitive ratio $r$ if for any sequence of queries $\mathbb{E}[\mathrm{ALG}(\sigma)] \le r \cdot \mathrm{OPT}(\sigma) + c$, where $c$ is a constant independent of $\sigma$, and the expectation is with respect to the internal randomness of the algorithm.
\end{remark}

\section{Competitive Ratio Analysis of the LRU Algorithm}\label{sec:LRU}
In this section, we consider the file-bundle caching problem with a single cache and show that the classic LRU algorithm still performs quite well among online deterministic algorithms. The LRU algorithm, which is well known since as early as in 1970\cite{mattson1970evaluation}, is a simple but powerful deterministic caching algorithm that works as follows: \emph{whenever a new page is requested, put it into the cache; if the cache is full so that the new page can not be fetched in, then evict the least recently used page from the cache}. It has been shown in \cite{sleator1985amortized} that: i) no deterministic online algorithm for the standard cache problem can achieve a competitive ratio better than $k$, where $k$ is the size of the cache, and ii) LRU has a competitive ratio of $k$. In the following, we show analogous results for the performance of the LRU algorithm in the case of the file-bundle caching problem.

\begin{theorem}\label{thm:deterministic-LRU}
	
	\begin{itemize}
		\item[i)] No deterministic online algorithm for the file-bundle caching problem can achieve a competitive ratio better than $k-l+1$, where $k$ is the size of the cache and $l$ is the length of the queries.
		\item[ii)] LRU has a competitive ratio of $k$.
\end{itemize}
\end{theorem}

\begin{proof}
	First, we note that by restricting the query sequence $\sigma$ always to contain the same $l-1$ pages, we can reduce the problem to the standard cache problem of a cache size $k-l+1$. The reason is that, for such a sequence of queries, $l-1$ locations of the cache are always occupied by the same pages. As a result, responding to a new query is equivalent to responding to a new single-page using the remaining $k-l+1$ locations in the cache. More precisely, defining $\mathcal{P} := \{\sigma| \{1,2,\cdots,l-1\}\subset Q_t,\forall t\}$, for any online deterministic algorithm $\mathrm{ALG}$, we have,
	\begin{align}\nonumber
	\sup_{\sigma}\frac{\mathrm{ALG}(\sigma)}{\mathrm{OPT}(\sigma)} \geq \sup_{\sigma\in\mathcal{P}}\frac{\mathrm{ALG}(\sigma)}{\mathrm{OPT}(\sigma)} \geq k-l+1,
	\end{align}
	where the second inequality is due to the fact that in the case of a standard cache problem with a cache size $k-l+1$, the best achievable competitive ratio for any deterministic online algorithm is $k-l+1$.
	
	Conversely, for an arbitrary sequence of requests $\sigma$, let us divide this sequence into several disjoint phases as follows:
	\begin{itemize}
		\item Phase $1$ begins at the first page of $\sigma$;
		\item Phase $i$ begins at the first query which contains the $(k+1)$-th distinct page since phase $i-1$ has begun (see, Example \ref{eq:phase}).
	\end{itemize}
	We show that the optimal offline algorithm OPT makes at least one cache miss each time a new phase begins. Denote the $j$-th query in phase $i$ by $Q_j^i$, and let $g^i$ be the total number of queries in phase $i$. For every $i$, consider queries $Q_2^i,Q_3^i,\cdots,Q_{g^i}^i$ and $Q_1^{i+1}$. Note that at the end of the query $Q_1^i$, the cache must contain all the pages requested by $Q_1^i$, and according to our definition of phases, there are all altogether at least $k+1$ distinct pages requested by queries $Q_1^i,Q_2^i,\cdots,Q_{g^i}^i,Q_1^{i+1}$. Thus, at least one cache miss will occur when the cache tries to answer queries $Q_2^i,Q_3^i,\cdots,Q_{g^i}^i, Q_1^{i+1}$.
	
	On the other hand, we show that LRU makes at most $k$ cache misses during each phase. According to the definition of phases, at most $k$ distinct pages are being requested in a phase, and we consider their first presence in the phase. Notice that according to the definition of LRU, if a query does not contain the first presence of any of these pages, then it does not incur a cache miss because, after its first presence, a page is kept in the cache until after the current phase ends. We want to show that there are at most $k$ queries in the phase that contain such first presences. In fact, every such query must contain at least one of these first presences, this in view of the fact that there are at most $k$ such first presences implies that there can be at most $k$ queries that contain such first presences. Thus, LRU makes at most $k$ cache misses during each phase. Finally, denoting the total number of phases by $K$, we have $\mbox{OPT}(\sigma)\ge K-1$ and $\mbox{LRU}(\sigma) \le kK$, which shows that LRU is $k$-competitive.
\end{proof}

\smallskip
\begin{example}\label{eq:phase}
	Consider a file-bundle caching problem with cache size $k=3$ and query length $l=2$. Moreover, consider a sequence of $T=6$ queries: 
	\begin{align}\nonumber
	\sigma = \{(4,1),(2,1),(2,1),(5,3),(4,3),(3,1),(2,3)\}.
	\end{align} 
	Dividing this sequence into phases we have, 
	\begin{align}\nonumber
	\underbrace{\Big((4\quad1)(2\quad1)(2\quad1)\Big)}_{\text{Phase} \,1} \underbrace{\Big((5\quad3)(4\quad3)\Big)}_{\text{Phase} \,2} \underbrace{\Big((3\quad1)(2\quad3)\Big)}_{\text{Phase} \,3}
	\end{align}
	Note that by definition, Phase 2 begins at the fourth query since it contains page $5$, which is the fourth distinct page since Phase 1 began. Similarly, Phase 3 begins at the sixth query since this query contains page $1$, which is the fourth distinct page after Phase 2 has begun. 
\end{example}

\section{Generalized $(h,k)$-Paging Problem}\label{sec:h,k}
In this section, we consider the generalized $(h,k)$-paging problem, i.e., $(h,k)$-paging problem in the file-bundle setting, where we compare the performance of our online algorithm with a cache sized $k$ with the optimal offline algorithm using a cache size $h\le k$. We first introduce the classic randomized marking algorithm. We then introduce a slightly modified version of the marking algorithm that better suit our problem, and then upper bound its competitive ratio under the generalized $(h,k)$-paging setting. We also give an almost matching lower bound for the competitive ratio of any randomized algorithm for the generalized $(h,k)$-paging problem.

\subsection{Marking Algorithm}
Marking algorithm was the first randomized caching algorithm introduced in \cite{fiat1991competitive}, where it was shown that for $h=k$, it achieves a competitive ratio of $2H_k$ against an oblivious adversary who generates the sequence of requested pages.

\begin{algorithm}\caption{Marking Algorithm \cite{fiat1991competitive}}\label{alg:standard-marking}
	Associate with each page in the cache one bit. If a cache page is recently used, the corresponding bit value is $1$ (marked); otherwise, the bit value is $0$ (unmarked). Initially, all cache pages are unmarked. Whenever a page is requested:
	\begin{itemize}
		\item If the page is in the cache, mark the page.
		\item Otherwise, if there is at least one unmarked page in the cache, evict an unmarked page uniformly at random, insert in the requested page, and mark it.
		\item Otherwise, unmark all the pages and start a new phase.
	\end{itemize}
\end{algorithm}

It was shown in \cite{young1991line} that for the standard $(h,k)$-paging problem, the above marking algorithm has a competitive ratio of no more than $2(\ln \frac{k}{k-h}-\ln \ln \frac{k}{k-h}+\frac{1}{2})$, when $\frac{k}{k-h}\ge e$, and $2$ otherwise. In the following, we introduce a slightly modified version of the marking algorithm and leverage this result to obtain a similar competitive ratio for the modified marking algorithm in the case of generalized $(h,k)$-paging problem. For this purpose, let us again consider the same phase partitioning of the request sequence $\sigma$:
\begin{itemize}
	\item Phase 1 begins at the first page of $\sigma$;
	\item Phase $i$ begins at the first query which contains the $(k+1)$-th distinct page since phase $i-1$ has begun.
\end{itemize}
We state the modified marking algorithm here.

\begin{algorithm}\caption{Query-wise Marking Algorithm}\label{alg:query-wise}
	Associate with each page in the cache one bit. If a cache page is recently used, the corresponding bit value is $1$ (marked); otherwise, the bit value is $0$ (unmarked). Initially, all cache pages are unmarked. Whenever a \emph{query} is requested
	\begin{itemize}
		\item If all the pages requested by the query are in the cache, mark those pages.
		\item Otherwise, if the requested query is not the first query of a phase (and thus the number of unmarked pages in the cache is greater than the number of requested pages not in the cache), then insert the newly requested pages into the cache by evicting enough unmarked pages uniformly at random, and mark the inserted pages.
		\item Otherwise, if the requested query is at the beginning of a phase, unmark all the pages and start the new phase.
	\end{itemize}
\end{algorithm}
Note that the main difference between the standard Marking Algorithm \ref{alg:standard-marking} and the Query-wise Marking Algorithm \ref{alg:query-wise} is that the latter unmarks all the pages at the end of each phase regardless of whether all the pages are marked or not.

\smallskip
\begin{definition}
	We call the \emph{new pages} in Phase $i$ as those pages that were not requested in Phase $i-1$, and the \emph{old pages} in Phase $i$ as those pages that were requested before in Phase $i-1$.
\end{definition}

\begin{lemma}\label{lem1}
	Let $k$ be the size of the cache, $m$ be the number of new pages in a phase, and $l$ be the length of the queries. Then the expected number of cache misses of the modified marking algorithm during the phase is at most $m+m\sum_{\ell=m+1}^{k}\frac{1}{\ell}.$
\end{lemma}

\begin{proof}
	Consider an arbitrary but fixed phase and let $X$ be a random variable denoting the number of cache misses by the query-wise marking algorithm during that phase. Moreover, let $Q_j$ be any query in the phase for which a cache miss occurs. Then either (I) $Q_j$ contains the first presence of at least one of the new pages, or (II) $Q_j$ \emph{does not} contain the first presence of any new page, but it contains the first presence of some $s_j\ge 1$ of the \emph{old} pages. Let $Q_{n_1},Q_{n_2},\cdots,Q_{n_p}$ be the queries satisfying condition (I), and $Q_{o_1},Q_{o_2},\cdots,Q_{o_q}$ be the queries satisfying condition (II). In the latter case, we also denote the number of old pages which were \emph{first} requested in queries $Q_{o_1},\cdots,Q_{o_q}$ by $s_{o_1},\cdots,s_{o_q}$, respectively. Clearly we have $p \leq m$.
	
	Now without loss of generality let us assume that $n_1<\cdots<n_p<o_1<\cdots<o_q$.\footnote{In fact, it is easy to argue that this assumption maximizes the expected number of cache misses by the algorithm.} We can write,
	\begin{align}\nonumber
	\mathbb{E}[X] &= \sum_{i=1}^{p}\mathbb{P}\{Q_{n_i} \text{suffers a cache miss}\} \cr 
	&\qquad+ \sum_{j=1}^{q}\mathbb{P}\{Q_{o_j} \text{ suffers a cache miss}\}\cr
	&= p + \sum_{j=1}^{q}\mathbb{P}\{Q_{o_j} \text{ suffers a cache miss}\}.
	\end{align}
	Next, we proceed to compute $\mathbb{P}\{Q_{o_j} \text{ suffers a cache miss}\}$. Let $A_{o_j}$ and $B_{o_j}=k-A_{o_j}$ be the number of marked and unmarked pages in the cache when query $Q_{o_j}$ comes. As queries $Q_{n_1},Q_{n_2},\cdots,Q_{n_p}$ altogether should mark at least $p$ pages in the cache, we must have
	$$B_{o_1} \le k-p.$$
	Also we note that $B_{o_{j+1}}=B_{o_j}-s_{o_j},\forall j$, and by the definition of a phase $B_{o_q}\ge s_{o_q}$, so that $B_{o_{q+1}}:= B_{o_q}-s_{o_q}\ge 0$.
	
	When query $Q_{o_j}$ comes, $B_{o_j}$ locations of the cache are still unmarked. However, the set of candidate pages in those unmarked locations is among the $k$ old pages in the cache at the beginning of the phase and should not be among the $A_{o_j}$ marked pages. Since the $A_{o_j}$ marked pages contain exactly $m$ new pages and $A_{o_j}-m$ old pages, the number of these page candidates is $k-(A_{o_j}-m)=B_{o_j}+m$. As $Q_{o_j}$ contains $s_{o_j}$ old pages which were requested in the previous phase but not yet in the current phase, the probability that $Q_{o_j}$ hits the cache is equal to the probability that all of these $s_{o_j}$ newly requested pages are among the $B_{o_j}$ unmarked pages in the cache. As by Algorithm \ref{alg:query-wise} all of the $B_{o_j}+m$ page candidates are equally likely to be in the cache (i.e., to be among the $B_{o_j}$ unmarked pages in the cache). Thus the probability that $Q_{o_j}$ hits the cache is equal to, 
	\begin{align*}\nonumber
	\mathbb{P}\{Q_{o_j} \text{hits the cache}\}&=\frac{{B_{o_j}\choose s_{o_j}}}{{B_{o_j}+m\choose s_{o_j}}}\\
	&=\prod_{i=1}^{s_{o_j}}(1-\frac{m}{B_{o_j}+m-s_{o_j}+i})\cr 
	&\ge 1-\sum_{i=1}^{s_{o_j}} \frac{m}{B_{o_j}+m-s_{o_j}+i}\\
	&=1-m\!\sum^{B_{o_j}}_{\ell=1+B_{o_{j+1}}}\frac{1}{\ell+m},
	\end{align*}
	where the inequality holds by $\prod_{i=1}^{n}(1-x_i)\ge 1-\sum_{i=1}^{n}x_i$, for every $0\le x_i\le1$. Now we can write,
	\begin{align}\label{eq:harmonic-bound}
	\mathbb{E}[X] &= p + \sum_{j=1}^{q}\mathbb{P}\{Q_{o_j} \text{suffers a cache miss}\}\cr
	&= p+ \sum_{j=1}^{q}(1-\mathbb{P}\{Q_{o_j} \text{hits the cache}\})\cr 
	&\leq p+m\sum_{j=1}^{q}\sum_{\ell=1+B_{o_{j+1}}}^{B_{o_j}}\frac{1}{\ell+m} \cr
	&\le p+ m\cdot \sum_{\ell=1}^{B_{o_1}}\frac{1}{\ell+m}.
	\end{align}
	Finally, since $B_{o_1} \le k-p$ and $p \le m$, the right hand side of \eqref{eq:harmonic-bound} is maximized for $p=m, B_{o_1}= k-m$. Therefore, we have $\mathbb{E}[X]\leq m+m\sum_{\ell=m+1}^{k}\frac{1}{\ell}$.
\end{proof}

\smallskip
\begin{theorem}\label{thrm1}
	The competitive ratio of the modified marking algorithm for the generalized $(h,k)$-paging problem is at most
	\begin{align}\nonumber
	r=2l (\ln \frac{k}{k-h}-\ln \ln \frac{k}{k-h}+\frac{1}{2}),
	\end{align}
	whenever $\frac{k}{k-h}\ge e$, and $r=2l$, otherwise.
\end{theorem} 
\begin{proof}
	Let $m_i$ denote the number of new pages in Phase $i$. Moreover, let $X_i$, and $X_i^{\mathrm{OPT}}$ be the number of cache misses that the marking algorithm and the optimal algorithm makes during Phase $i$, respectively. From Lemma \ref{lem1}, we have
	\begin{equation*}
		\mathbb{E}[X_i]\leq m_i+m_i\sum_{\ell=m_i+1}^{k}\frac{1}{\ell} \leq m_i+m_i\ln\frac{k}{m_i}.
		\end{equation*}
	Also, since in total there are $k +m_i$ distinct pages in Phases $i -1$ and $i$, we have $h+l(X_{i-1}^{\mathrm{OPT}}+X_i^{\mathrm{OPT}}) \ge k +m_i$, where the left-hand is an upper bound on the maximum number of distinct page requests during Phases $i-1$ and $i$. Therefore,
	\begin{equation*}
	X_{i-1}^{\mathrm{OPT}}+X_i^{\mathrm{OPT}} \ge \frac{k-h+m_i}{l}.
	\end{equation*}
	Now let $K$ be the total number of phases in a given sequence of queries $\sigma$. Then the cost of the modified marking algorithm and that of the optimal algorithm can be bounded by
	\begin{align}\nonumber
	&\mathbb{E}[\mathrm{MARK}(\sigma)] = \sum_{i = 1}^{K}E[X_i]\le\sum_{i = 1}^{K}(m_i+m_i\ln\frac{k}{m_i}),\cr 
	&\mathrm{OPT}(\sigma) = \sum_{i=1}^{K}X_i^{\mathrm{OPT}}\ge\frac{1}{2}\sum_{i=1}^{K} \frac{k-h+m_i}{l}-c.
	\end{align}
	As a result, the competitive ratio $r$ of the modified marking algorithm is at most $r \le 2\max_m T(m)$, where $T(m) := \frac{m+m\ln\frac{k}{m}}{\frac{k-h+m}{l}}$. Finally, it is shown in Appendix I (Lemma \ref{lemm:max-T}), that the maximum value of $T(m)$ over all ranges of $m$ is at most $l\cdot (\ln \frac{k}{k-h}-\ln \ln \frac{k}{k-h}+\frac{1}{2})$ if $\frac{k}{k-h}\ge e$, and it is $l$, otherwise. This completes the proof.   
\end{proof}

\subsection{A Lower Bound for the Competitive Ratio}
In this section, we give a lower bound on the competitive ratio of any randomized algorithm for the generalized $(h,k)$-paging problem.

\smallskip
\begin{theorem}\label{thrm2}
	Let $\frac{k}{k-h}\ge e$. Then for any randomized algorithm ALG, there is a sequence of queries $\sigma$ such that, 
	\begin{align}\nonumber
	\frac{\mathbb{E}[\mathrm{ALG}(\sigma)]}{\mathrm{OPT}(\sigma)}\ge l\ln \frac{k\!+\!1}{k\!-\!h\!+\!l\!+\!1}\!-l\ln\ln\frac{k\!+\!1}{k\!-\!h\!+\!l\!+\!1}\!-\!2l\!+\!1.
	\end{align}
\end{theorem}
\medskip

\begin{proof}
	As before, let us denote the content of cache for ALG at time $t$ by $C_t\subseteq [N]$, where $|C_t|= k$. Moreover, let $D = \{D_{\e}\subseteq[N]:|D_{\e}|=k\}$ be the set of all possible states for $C_t$. Then any randomized algorithm can be described using a probability transition matrix from $D$ to itself given the query sequence $\sigma$. Consider an adversary that at each time $t$ knows the probability distribution induced by the ALG over the set of cache states. In other words, at every time $t$ the adversary knows $P_{\e}^t:=\mathbb{P}\{C_t=D_{\e}\}, \forall D_{\e}\subseteq D$, where $\sum_{D_{\e}\subseteq D}P_{\e}^t=1$. Now let $P_i^t = \sum_{i\in D_{\e}}P_{\e}^t$ be the probability that page $i$ is in $C_t$. We have,
	\begin{align*}
	\sum_{i=1}^{N}P_i^t &= \sum_{i=1}^{N}\sum_{i\in D_{\e}}P_{\e}^t=\!\!\sum_{D_{\e}\subseteq D}\!\sum_{i=1}^{N}\mathbbm{1}_{\{i\in D_{\e}\}}P_{\e}^t\!=\!\!\sum_{D_{\e}\subseteq D}\!\!kP_{\e}^t=k.
	\end{align*}
	
	Next, we generate a sequence of queries in a segment as follows: Let $m$ be a positive integer to be determined later, and $\d \rightarrow0^+$. During each segment,
	\begin{enumerate}
		\item[(1)] Generate $\frac{k-h+m}{l}$ queries requesting $k-h+m$ `new' pages, i.e. pages not in the the cache of ALG and OPT.\footnote{Note that this definition of `new' pages is different from that given in the proof of Lemma \ref{lem1}.} (Since $N$ is large we can always find such `new' pages.)
		\item[(2)] Denote the $k-h+m$ `new' pages together with the $h$ pages in the cache of OPT before the beginning of the segment to be $k+m$ \emph{candidate} pages. Generate another query requesting $l$ `new' pages, pick $l-1$ of them and denote them as \emph{fixed} pages.
		\item[(3)] For $i=1,\ldots,h-l$, we generate a query that consists of $l-1$ fixed pages and one page $p_i$ from the candidate pages that are least likely to be contained in the cache of ALG. That is, given time $t$ and for $i=1,\ldots,h-l$, while $\exists j\le i-1,$ such that $P_{p_j}^t\le 1-\d $, generate a query requesting $p_j$ and $l-1$ fixed pages. Since $\forall j\le i-1,P_{p_j}^t\ge 1-\d$, we have
		\begin{align}\nonumber
		&\sum_{\{\text{candidate pages}\ s\}}P_s^t\le \sum_{i=1}^{N}P_i^t=k,\cr
		&\qquad\sum_{j=1}^{i-1} P_{p_j}^t \ge (i-1)(1-\d).
		\end{align}
		As the total number of candidate pages is $k+m$, by pigeonhole principle there exists a candidate page $p_i$ such that $P_{p_i}^t\le \frac{k-(i-1)(1-\d )}{k+m-(i-1)}$. We generate a query requesting $p_i$ and $l-1$ fixed pages.
	\end{enumerate}
	Now we compute the number of cache misses during each segment for ALG and OPT. Since the total number of candidate pages requested during stage (3) is $h-l$, the optimal algorithm can keep these $h-l$ candidate pages in the cache all through stages (2) and (3) and only use the remaining $l$ pages in the cache to respond to queries generated in (1) and (2). Therefore the number of cache misses for the OPT in each of the above three states are:
	\begin{itemize}
		\item In (1), $\frac{k-h+m}{l}$ cache misses.
		\item In (2), $1$ cache miss.
		\item In (3), $0$ cache miss.
	\end{itemize}
	On the other hand, the number of cache misses for the ALG are given by
	\begin{itemize}
		\item In (1), $\frac{k-h+m}{l}$ cache misses.
		\item In (2), $1$ cache miss.
		\item In (3), the number of cache misses is at least 
		\begin{align*}
		&\sum_{i=1}^{h-l}(1-\frac{k-(i-1)(1-\d )}{k+m-(i-1)})\\
		&\qquad=\sum_{i=1}^{h-l}\frac{m-\d(i-1)}{k+m-i+1}\\
		&\qquad=m\cdot\sum_{i=1}^{h-l}\frac{1}{k+m-i+1}\qquad \text{(because $\d \rightarrow 0^+$)}\\
		&\qquad\ge m\cdot\ln\frac{k+m+1}{k+m-h+l+1}
		\end{align*}
	\end{itemize}
	As a result, the competitive ration of the ALG over this generated sequence of queries is at least
	\begin{align*}
	r\ge& \frac{\frac{k-h+m}{l}+1+m\cdot\ln\frac{k+m+1}{k+m-h+l+1}}{\frac{k-h+m}{l}+1}\\
	=&1+l\cdot \frac{m}{k-h+m+l}\ln\frac{k+m+1}{k-h+m+l+1}\\
	\ge &1+l\cdot\frac{m}{k-h+l+1+m}\ln\frac{k+1}{k-h+l+1+m}\\
	=&1+l\cdot f(m),
	\end{align*}
	Finally, maximizing the right-hand side of the above inequality with respect to $m$ (see, Lemma \ref{eq:lower-bound} in Appendix I) we obtain the desired lower bound on the competitive ratio.
\end{proof}

\begin{remark}
	Comparing the results of Theorems \ref{thrm1} and \ref{thrm2}, one can see that when the size of queries is much smaller than the size of the cache, i.e., $l\ll h, k$, (which is indeed the case in almost all practical applications), we have,
	\begin{align*}
	&2l\cdot (\ln \frac{k}{k-h}-\ln \ln \frac{k}{k-h}+\frac{1}{2})\\
	&\approx 2\cdot(l\ln\frac{k+1}{k-h+l+1}-l\ln\ln\frac{k+1}{k-h+l+1}-2l+1).
	\end{align*}
	As a result, the competitive ratio of the marking algorithm is approximately within a factor of two of the optimal one.
\end{remark}

\section{Distributed File-Bundle Caching}\label{sec:distributed}
Now we turn our attention to the distributed version of the file-bundle caching problem. This section considers the organization of the caches of geographically distributed data centers, each keeping a replication of the whole data set, in order to answer file-bundle queries. In this section, we propose two effective distributed caching algorithms that use $m=l+1$ caches and achieve small competitive ratios. Building upon this, we then propose a more general framework for designing distributed caching algorithms.

\subsection{Competitive Distributed Algorithms with $l+1$ Caches}
Here we propose a deterministic and a randomized distributed caching algorithm that can be viewed as a generalization of the LRU and the modified marking algorithms, respectively. The main results of this section are summarized in the following theorems:

\begin{theorem} \label{thrm3}
	For $m = l+1$ identical caches in the system, there is a deterministic distributed caching algorithm with a competitive ratio of $r=l^2 +l$. 
\end{theorem}

\begin{theorem} \label{thrm4}
	For $m = l+1$ identical caches in the system, there is a randomized distributed caching algorithm with a competitive ratio of $r=2l\cdot(\ln(2l+1)-\ln\ln l+\frac{1}{2})$. 
\end{theorem}

The key element in establishing the above results is to reduce the problem to that of a single-cache case by introducing a so-called \emph{virtual} cache. For this purpose, let $g=\floor{\frac{k}{l^2}}$, where we recall that $k$ is the size of every cache, and $l$ is the length of every query. Let us consider a virtual cache of size $k+gl$ whose content at time $t$ is denoted by $C^*_t$. Upon the arrival of a new query at time $t$, we first update the virtual cache from $C_t^*$ to $C_{t+1}^*$ using the earlier LRU or marking algorithms. Accordingly, we update the actual set of caches from $C_t=\{C^i_t\}_{i=1}^m$ to $C_{t+1}=\{C^i_{t+1}\}_{i=1}^m$ in such a way to mimic the responsive behavior of the virtual cache with respect to the future queries.

\begin{algorithm}\caption{Multiple to Single Cache Reduction}\label{alg:reduction} 
	Maintain a virtual cache of size $k+gl$, and denote its content at time $t$ by $C_t^* = \{c_t^1,c_t^2,\cdots,c_t^{k+gl}\}$. After query $Q_t$ arrives, update $C_t^*$ to $C_{t+1}^*$ using the LRU/marking algorithm. Distribute the content of the updated virtual cache $C^*_{t+1}$ among $m=l+1$ caches of size $k$ as follows:
	\begin{equation*}
	\begin{split}
	C_{t+1}^1 &= C_{t+1}^* - \{c_{t+1}^1,c_{t+1}^2,\cdots,c_{t+1}^{gl}\},\\
	C_{t+1}^2 &= C_{t+1}^* - \{c_{t+1}^{gl+1},c_{t+1}^{gl+2},\cdots,c_{t+1}^{2gl}\},\\
	\vdots&\\
	C_{t+1}^{l+1} &=C_{t+1}^* - \{c_{t+1}^{gl^2+1},c_{t+1}^{gl^2+2},\cdots,c_{t+1}^{gl^2+gl}\}.\\
	\end{split}
	\end{equation*}
	Note that since $gl^2+gl \leq k+gl$, the above construction is well-defined. 
\end{algorithm}

Next, we show that following Algorithm \ref{alg:reduction} the ability to respond to future queries using either actual set of caches or the virtual cache is indeed the same.

\begin{lemma} \label{thrm5}
	For every query $Q_t\subset\mathcal{O},\ |Q_t|=l$, we have
	$$f(C_t,Q_t) = f(C_t^*,Q_t),$$
	where $f(C_t,Q_t)\!=\!\prod_{i=1}^{l+1} \mathbbm{1}_{\{Q_t\not\subseteq C_t^i\}}, f(C_t^*,Q_t)\!=\!\mathbbm{1}_{\{Q_t\not\subseteq C_t^*\}}$.
\end{lemma}

\begin{proof}
	First we show that $f(C_t,Q_t) \ge f(C_t^*,Q_t)$. Otherwise, if $f(C_t,Q_t) < f(C_t^*,Q_t)$, this means that $f(C_t,Q_t)=0$ and $f(C_t^*,Q_t)=1$. But in that case,
	\begin{equation*}
	\begin{split}
	f(C_t^*,Q_t) =1 &\Rightarrow  Q_t\not\subseteq C_t^*\\
	&\Rightarrow  \exists q\in Q_t, q \not\in C_t^*\\
	&\Rightarrow \forall i \in\{1,2,\cdots,l+1\}, q \not\in C_t^i\\
	&\Rightarrow \forall i \in\{1,2,\cdots,l+1\},\mathbbm{1}\{Q_t\not\subseteq C_t^i\}=1\\
	&\Rightarrow f(C_t,Q_t)= \prod_{i=1}^{l+1} \mathbbm{1}\{Q_t\not\subseteq C_t^i\}=1,
	\end{split}
	\end{equation*}
	where the third line follows from the construction of $C_t$, as we have $C_t^i \subseteq C_t^*, \forall i$. This contradiction shows that $f(C_t,Q_t) \ge f(C_t^*,Q_t)$.
	
	Next, we show that $f(C_t,Q_t) \le f(C_t^*,Q_t)$. Otherwise, if $f(C_t,Q_t) > f(C_t^*,Q_t)$, then $f(C_t^*,Q_t)=0$ and $f(C_t,Q_t)=1$. This means that $Q_t\subseteq C_t^*$, and $Q_t\not\subseteq C_t^i =C_{t}^* - \{c_{t}^{(i-1)gl+1},c_{t}^{(i-1)gl+2},\cdots,c_t^{igl}\} , \forall i$. Thus,
	\begin{align}\nonumber
	|Q_t\cap\{c_{t}^{(i-1)gl+1},c_{t}^{(i-1)gl+2},\cdots,c_t^{igl}\}|\ge 1,\forall i,
	\end{align}
	and we can write,
	\begin{equation*}
	\begin{split}
	|Q_t| &= |Q_t\cap C_t^*|\\
	&\ge|\cup_{i=1}^{l+1} (Q_t\cap\{c_{t}^{(i-1)gl+1},c_{t}^{(i-1)gl+2},\cdots,c_t^{igl}\})|\\
	&=\sum_{i=1}^{l+1}|Q_t\cap\{c_{t}^{(i-1)gl+1},c_{t}^{(i-1)gl+2},\cdots,c_t^{igl}\}|\ge l+1.
	\end{split}
	\end{equation*}
	This contradiction shows that $f(C_t,Q_t) \le f(C_t^*,Q_t)$, which completes the proof.
\end{proof}

Using the above lemma, we next proceed to prove Theorems \ref{thrm3} and \ref{thrm4}.

{\bf \emph{Proof of Theorem \ref{thrm3}}:} We use the same technique as in the proof of Theorem \ref{thm:deterministic-LRU}. As before, we divide the request sequence $\sigma$ into phases where Phase 1 begins at the first page of $\sigma$, and Phase $i$ begins at the first query, which contains the $(k+gl+1)$-th distinct page since Phase $i-1$ has begun. We show that OPT makes at least $g+1$ cache miss each time a new phase begins. Denote the $j$-th query in phase $i$ by $Q_j^i$, and let $g^i$ be the number of queries in Phase $i$. Consider queries $Q_2^i,Q_3^i,\cdots,Q_{g^i}^i$ and $Q_1^{i+1}$. Note that at the end of the query $Q_1^i$, the virtual cache must contain all the pages requested by $Q_1^i$, and according to our definition of phases, at least $k+gl+1$ distinct pages are requested by queries $Q_1^i,Q_2^i,\cdots,Q_{g^i}^i, Q_1^{i+1}$. Let $x$ be the number of cache misses of the OPT in answering queries $Q_2^i,Q_3^i,\cdots,Q_{g^i}^i, Q_1^{i+1}$. Then $xl+k\ge k+gl+1$, with the left-hand side being an upper bound for the maximum number of distinct pages in queries $Q_1^i,Q_2^i,\cdots,Q_{g^i}^i,Q_1^{i+1}$. Therefore, $x \ge g+1$. 

On the other hand, following the same argument as in the proof of Theorem \ref{thm:deterministic-LRU}, since at each time the virtual cache is updated based on LRU, we conclude that the virtual cache makes at most $k+gl-l+1$ cache misses during each phase. Thus denoting the number of phases by $K$, we have $\mathrm{OPT}(\sigma) \ge (K-1)(g+1)$ and $\mathrm{DLRU}(\sigma) \le (k+gl-l+1)K$. As $\frac{k+gl-l+1}{g+1} \le \frac{(g+1)l^2+(g+1)l}{g+1} = l^2+l$, therefore the distributed LRU algorithm has a competitive ratio of $l^2+l$. This in view of Lemma \ref{thrm5} completes the proof. \hfill $\blacksquare$

\smallskip
{\bf \emph{Proof of Theorem \ref{thrm4}}:} From Theorem \ref{thrm5} and Lemma \ref{thrm5} we know that the competitive ratio of the distributed marking algorithm is at most the competitive ratio of the modified marking algorithm for the generalized $(k,k+gl)$-paging problem (note that the size of the algorithm cache is now replaced by that of the virtual cache, i.e., $h=k+gl$). Therefore, using the result of Theorem \ref{thrm1}, we get:
\begin{align*}
r&\le 2l\cdot(\ln\frac{k+gl}{k+gl-k}-\ln\ln\frac{k+gl}{k+gl-k}+\frac{1}{2})\\
&\le 2l\cdot(\ln\frac{(g+1)l^2+gl}{gl}-\ln\ln\frac{(g+1)l^2+gl}{gl}+\frac{1}{2})\\
&\le 2l\cdot(\ln(2l+1)-\ln\ln l+\frac{1}{2}).
\end{align*}
\hfill $\blacksquare$

\subsection{Generalization to Arbitrary Number of Caches}
As we showed earlier, when there are $m=l+1$ caches in the system, the distributed caching algorithms that are based on virtual cache reduction will give us competitive algorithms. Here we attempt to generalize this idea to larger ranges of $m$. We follow the same idea as in Algorithm \ref{alg:reduction}, which is to reduce the distributed caching problem to a single-cache generalized $(h,k)$-paging problem, and then use the marking (or LRU) algorithm to solve it. This is done by introducing a single virtual cache of a larger size while making sure that the virtual cache and the actual caches perform the same in responding to any query. Unfortunately, when $m\neq l+1$, the construction given in Algorithm \ref{alg:reduction} is no longer valid. In that case, we can distribute our virtual cache's content among a set of $m$ caches using a generalized notion of $r$-dense families \cite{rodl1985packing}. Such construction will again preserve the responsive equivalence between the virtual and the actual caches. However, as we shall see, the number of caches needed to achieve a substantial decrease in the competitive ratio is prohibitively large, making this approach practically unrealizable. Nevertheless, we present our results here, as it may indicate a fundamental difficulty in the organization of geographically distributed caches to answer file-bundle queries.

\smallskip
\begin{definition}[$r$-dense families \cite{rodl1985packing}]\label{def4}
	Given positive integers $r < k < N$, let $\mathcal{F}$ be a family of $k$-element subsets of $\{1, 2, ... , N\}$. We say that $\mathcal{F}$ is $r$-dense if any $r$-element subset of $\{1, 2, ... , N\}$ is contained in at least one member of $\mathcal{F}$.
\end{definition}

\begin{lemma}(\cite{rodl1985packing})\label{lem2}
	Let $M(N, k, r)$ be the minimal number of elements of a $r$-dense family $\mathcal{F}$. Then,
	\begin{align}\label{eq:rdense}
	M(N,k,r) \le \frac{{N\choose r}}{{k\choose r}}(1+\log{k\choose r}).
	\end{align}
\end{lemma}

To see how to use $r$-dense families in our virtual cache reduction, using Definition \ref{def4} and Lemma \ref{lem2} to our problem setting we can maintain a virtual cache of size $k^*$ if we have a set of $m=M(k^*,k,l)\le\frac{{k^*\choose l}}{{k\choose l}}(1+\log{k\choose l})$ caches, each of size $k$. Therefore, selecting the contents of our $m$ actual caches to be $k$-element subsets of the virtual cache that form a $l$-dense family, we are guaranteed that the virtual cache and the actual ones respond identically to every query. This is because if a query of size $l$ misses the virtual cache, so do the actual caches (note that every actual cache is a subset of the virtual cache). On the other hand, if a query of size $l$ (viewed as a $l$-element subset) hits the virtual cache, since the collection of actual caches forms a $l$-dense family, so the query must be contained in at least one of the $k$-subsets (actual cache) of the virtual cache. Therefore, the requested query will also hit at least one of the actual caches. This establishes the equivalence between the virtual and the actual caches.

Now using the result of Theorem \ref{thrm1}, we know that applying the modified marking algorithm on the virtual cache when compared with an optimal offline benchmark of cache size $k$, the competitive ratio is at most,
\begin{align}\label{eq:general-distributed}
r=2l\cdot (\ln \frac{k^*}{k^*-k}-\ln \ln \frac{k^*}{k^*-k}+\frac{1}{2}).
\end{align}
This shows that the competitive ratio of the generalized distributed caching algorithm with $m=M(k^*,k,l)$ caches of size $k$ is at most \eqref{eq:general-distributed}.

	\begin{remark}
		To interpret the result, using approximation and neglecting less important terms, we can compute from equation (\ref{eq:rdense}) and (\ref{eq:general-distributed}) that for some $h <l$, if we want the randomized distributed algorithm to have a competitive ratio of $r = 2l\ln(\frac{l}{h})$, then a number of approximately $m\approx e^hl\log\frac{k}{l}$ caches are needed. The prohibitively large number makes this approach practically unrealizable. Although we do not give a lower bound on the competitive ratio of any distributed algorithm for file-bundle caching at this point, this negative result provides some insight into the difficulty of designing competitive distributed algorithms for the file-bundle caching problem. 
	\end{remark}

\section{Simulations}\label{sec:simulation}
We design a simulation model to explore the practical performance of our proposed algorithms. For the sake of comparison, we use two benchmark algorithms. One is \emph{random eviction}, which, whenever needed, will uniformly at random selects a file from the cache and evict it. The other is the \emph{farthest in future} (FF) algorithm, which, knowing all future queries, will evict the files that will be requested farthest in the future whenever evictions are needed. The FF algorithm is proved to be the optimal offline algorithm for the standard caching problem (i.e., when $l=1$)\cite{mattson1970evaluation}. In the file-bundle caching problem, the FF algorithm can be viewed as a greedy approach to the problem, i.e., always evict the most ''useless'' files in the cache, which are the files that will be requested farthest in the future. Therefore, although the FF algorithm may not be the optimal offline algorithm for file-bundle caching, as is shown in Theorem \ref{thrm6} (Appendix I), it provides a $2l$-approximation for the optimal offline algorithm. Thus, we use the FF algorithm as a proxy for the optimal offline algorithm for the file-bundle caching problem, mainly because computing the exact optimal offline algorithm can be very time-consuming\footnote{In fact, how to efficiently compute the optimal offline algorithm for file-bundle caching itself is an interesting open problem.}.

	\subsection{Workload Characterization and  Simulation Parameters}
	Although file-bundle is the mode of file requests in much data-intensive scientific analysis, most workload traces and logs maintained by scientific data centers are on a per-file basis and do not give information about requested file-bundles. Apart from that, there are even fewer efforts made to derive workload traces and logs of file-bundle caching activities in other environments such as web-caching.

In the absence of workload traces and logs available for experimentation, we construct simulated workloads consisting of sequences of queries each request a file-bundle selected from a large set of candidate files and test our proposed algorithms' performance on these simulated workloads. We choose parameters such that our simulated workload is as close as possible to observed real experiments that trace single file requests. We use \emph{miss ratio} as the performance measurement metrics. In the following, we summarize the major parameters of our simulation. During the experiments, these parameters are varied to observe the impact on the performance of the algorithms.

\paragraph{Popularity Distribution:}
The popularity distribution considers how single files are selected into the file-bundles requested by queries. We examine two distributions' effects: a uniformly random distribution, where all the candidate files have the same probability of being selected by any query, and a Zipf distribution. In a Zipf distribution, the probability of selecting the $i^{th}$ most popular element is proportional to $\frac{1}{i^s}$, with $s$ being the value of the exponent characterizing the distribution. In our experiments, we first generate a number of candidate queries with a fixed length. Each of these candidate queries is generated independently and identically, and the probability of presence for each file in any particular candidate query follows a Zipf distribution. We then generate the real queries from these candidate queries such that the probability that a candidate query to be among the real queries obeys another Zipf distribution. In the experiments, the exponent value, $s$, for both Zipf distributions is set to be $1$.

\paragraph{Size of the Cache and Length of the Queries:}
The size of the cache $k$ and the length of the queries $l$ are the two most important parameters in our theoretical analysis, and perhaps they are also the parameters that have the most impact on the practical performance of the file-bundle caching algorithms. Naturally, we will expect increasing $k$ will reduce the miss ratio while increasing $l$ will increase the miss ratio.

\paragraph{Size of the data set:}
The size of the data set corresponds to the total number of candidate files in our simulation experiments. Although less discussed in the theoretical analysis above, the size of the data set can substantially impact the performance of the algorithms. Consider two extreme cases: when the size of the data set is very small, say less than or equal to the size of the cache, then we can store all the files in the cache and never incur any cache miss; on the other hand, if the data set is extremely large and all the files are approximately equally likely to be selected into the request, then we would expect even the offline optimal algorithm to have a miss ratio of nearly $100\%$ since every query is very likely to contain a ``new'' file, i.e., a file that is never requested before.

	\subsection{Simulation Results}
	This section presents some simulation results of our proposed algorithms to show the impact of some representative parameters such as varying cache size, varying query length, varying size of the data set, and different request distributions. In each simulation run, the \emph{miss ratio} is calculated by running the algorithm on a sequence of $40,000$ queries.\footnote{All the simulations are performed in Matlab and are run by a computer with 3.1 GHz Intel Core i5 CPU and 8GB memory.}

The effect of varying cache size of $k$ is summarized in figure \ref{fig1000}. In this set of experiments, we set the size of the data set to be $N=10,000$, the length of the queries to be $l=10$. In figure \ref{fig1001}, we let the requested files follow the Zipf distribution, and in figure \ref{fig1002} we let the requested files follow the random distribution. The LRU algorithm outperforms the query-wise marking algorithm, and as the size of the cache increases, the \emph{miss ratio} of all the algorithms decreases. When the requested files follow the Zipf distribution, the performance gap between the LRU algorithm, query-wise marking algorithm, and the FF algorithm shrinks as $k$ increases. This is because as $k$ increases, all the three algorithms can handle frequently seen queries quite well, and for rarely seen queries, none of the algorithms can handle them well. On the other hand, when the requested files follow the random distribution, the performance gap remains approximately the same because there are no frequently seen queries nor rarely seen queries.
\begin{figure} 
	\subfigure[Zipf Distribution]{ 
		\includegraphics[width=0.48\linewidth]{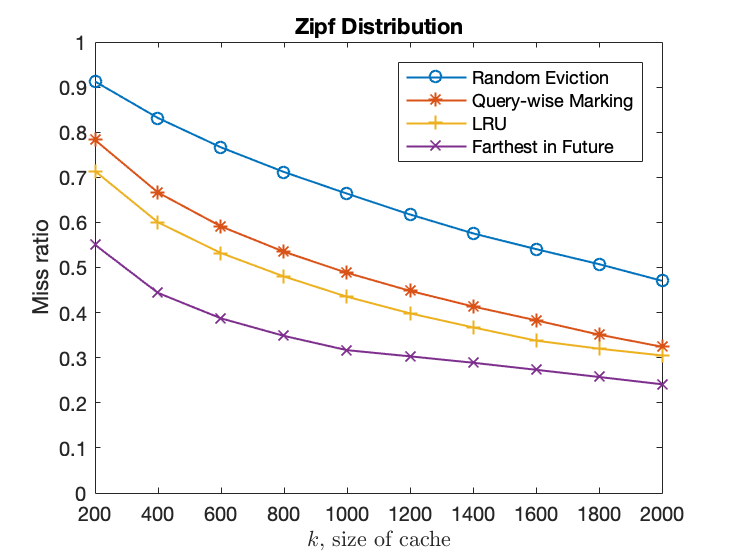}\label{fig1001}}
	\subfigure[Random Distribution]{
		\includegraphics[width=0.48\linewidth]{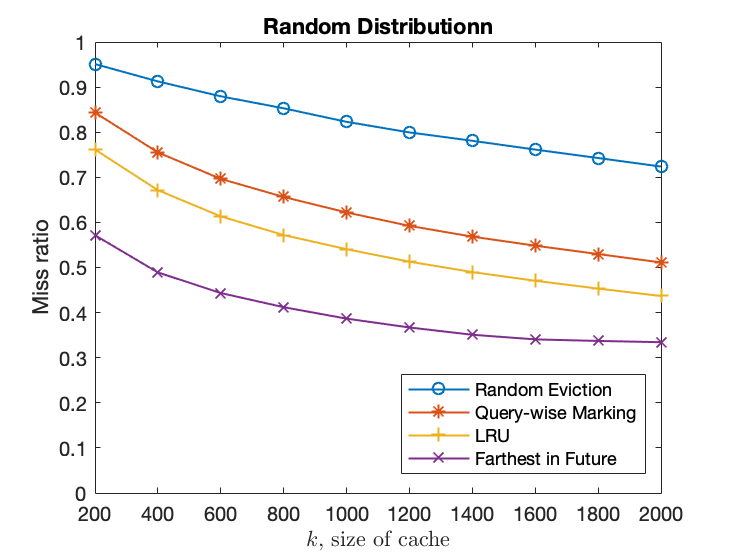}\label{fig1002}}
	\caption{Effect of varying cache size $k$ with data set size $N=10,000$ and query length $l=10$}\label{fig1000}
\end{figure}  

The effect of varying query length $l$ is summarized in Figure \ref{fig2000}. In this set of experiments, we set the size of the data set to be $N=10,000$, the cache's size to be $k=600$. In Figure \ref{fig2001}, we let the requested files follow the Zipf distribution, and in figure \ref{fig2002} we let the requested files follow the random distribution. As the length of the queries increases, the \emph{miss ratio} of all the algorithms increases. In fact, from the graph, we can see that the effect of increasing query length $l$ with fixed cache size $k$ is similar to that of decreasing cache size $k$ with fixed query length $l$.

\begin{figure} 
	\subfigure[Zipf Distribution]{ 
		\includegraphics[width=0.48\linewidth]{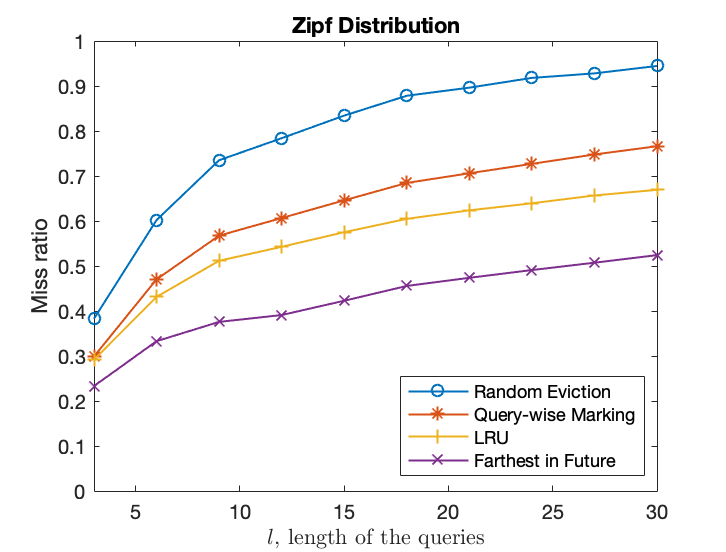}\label{fig2001}}
	\subfigure[Random Distribution]{
		\includegraphics[width=0.48\linewidth]{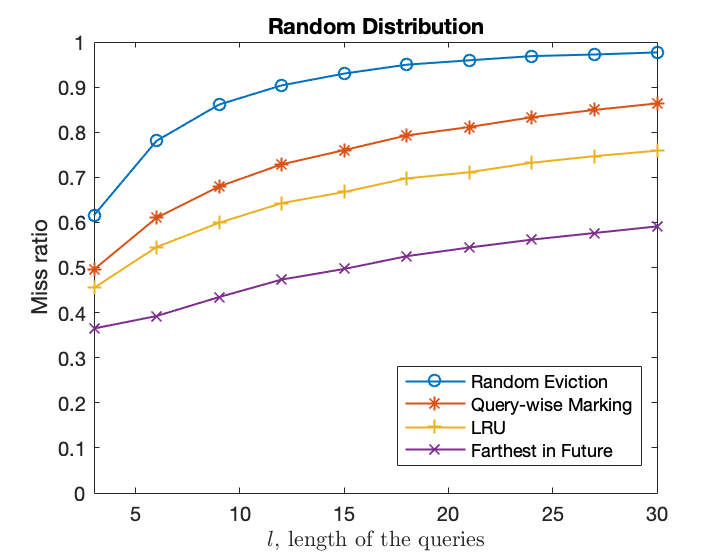}\label{fig2002}}
	\caption{ Effect of varying query length $l$ with data set size $N=10,000$ and cache size $k=600$}\label{fig2000}
\end{figure}  

The effect of varying data set size of $N$ is summarized in figure \ref{fig3000}. In this set of experiments, we set the size of the cache to be $k=800$, the length of the queries to be $l=10$. In Figure \ref{fig3001}, we let the requested files follow the Zipf distribution, and in Figure \ref{fig3002}, we let the requested files follow the random distribution. We run simulations for data set size $N = 800, 1000, 2000, 5000, 10000$, $20000, 50000, 100000$. From the graph, we can see that the \emph{miss ratio} of the algorithms increases dramatically as $N$ increases when $N\le 10000$, and becomes stable after that. We conclude that when the size of the cache $k$ is comparable to the size of the data set $N$, their relative size can substantially affect the performance of file-bundle caching algorithms. However, when $k<<N$, then their relative size will not matter so much. Instead, the relative size of the cache size $k$ and the query length $l$ dominates the performance of file-bundle caching algorithms.
\begin{figure} 
	\subfigure[Zipf Distribution]{ 
		\includegraphics[width=0.48\linewidth]{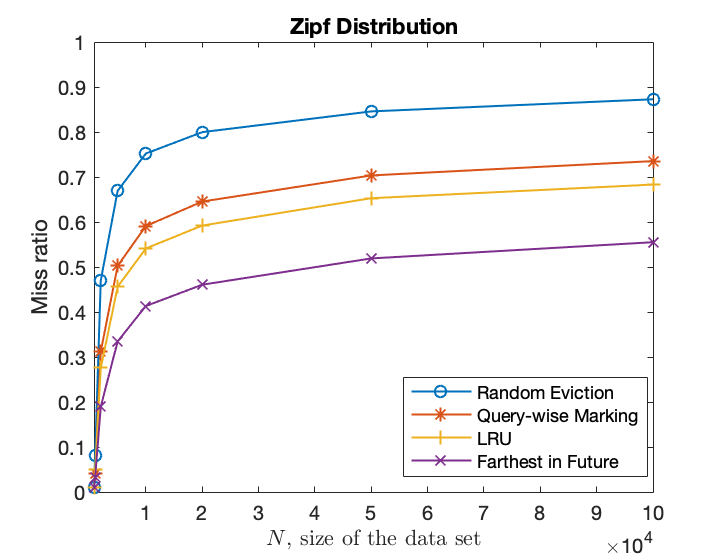}\label{fig3001}}
	\subfigure[Random Distribution]{
		\includegraphics[width=0.48\linewidth]{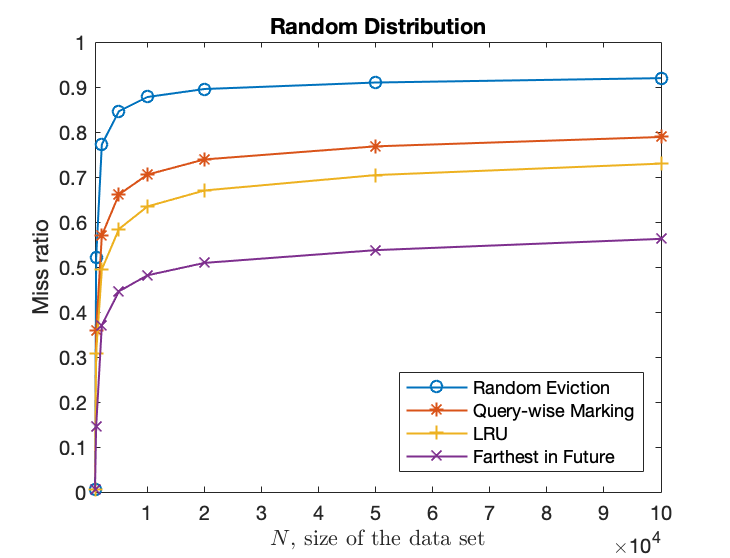}\label{fig3002}}
	\caption{Effect of varying size of the data set $N$  with query length $l=10$ and cache size $k=800$}\label{fig3000}
\end{figure}  

The performance evaluation of our distributed caching algorithm is shown in Figure \ref{fig4000}. In this set of experiments, we set the size of the data set to be $N=10,000$. We vary the length of the queries $l$, and with every $l$, the performance of a set of $m=l+1$ caches with the same cache size $k = 80l$ is tested using the distributed caching algorithm of Theorem \ref{thrm4}. For comparison, the performance of a monolithic cache with cache size $k' = 160l$ is also tested using both the query-wise marking algorithm and the FF algorithm. In Figure \ref{fig4001}, we let the requested files follow the Zipf distribution, and in Figure \ref{fig4002}, we let the requested files follow the random distribution. We run simulations for query length $l = 4,8,\ldots,40$. From the graph, we can see a performance gap between distributed caching and a monolithic cache, and that gap becomes greater as we increase $l$, especially when requested files follow the Zipf distribution. This suggests that for large $l$ distributed caching may not be very useful, as the cost increases (the number of caches needed is $m=l+1$) but the improvement of performance decreases. As a by-product, if we focus on the monolithic cache's performance, we can see that as long as the ratio of $\frac{k}{l}$ is fixed, the performance of the file-bundle caching algorithms is relatively stable with varying $l$.
\begin{figure} 
	\subfigure[Zipf Distribution]{ 
		\includegraphics[width=0.48\linewidth]{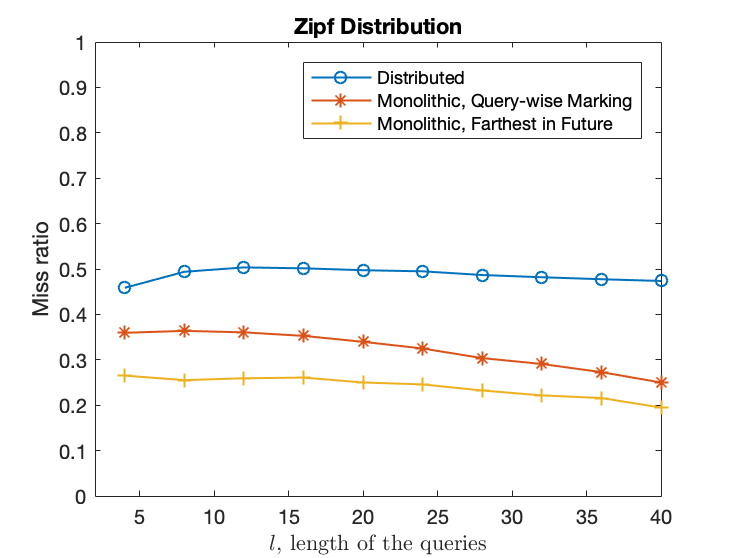}\label{fig4001}}
	\subfigure[Random Distribution]{
		\includegraphics[width=0.48\linewidth]{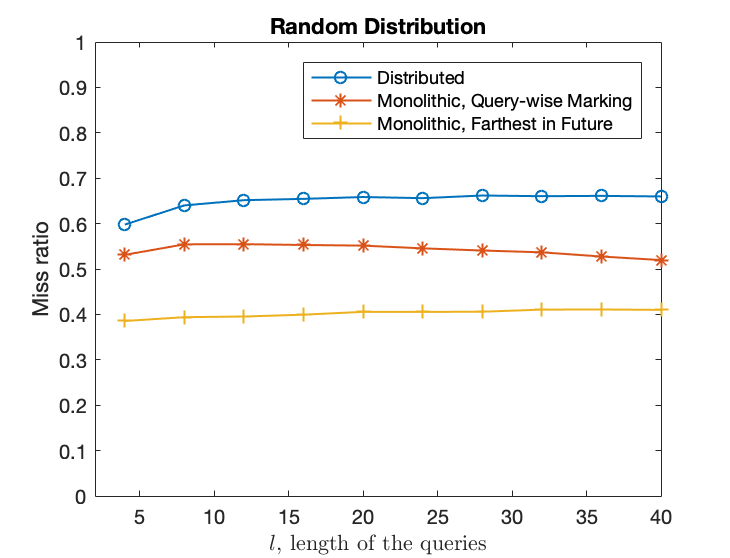}\label{fig4002}}
	\caption{Evaluation of distributed file-bundle caching algorithm}\label{fig4000}
\end{figure}  

Our last experiment considers adversarially generated queries. In this experiment, we set the size of the data set to be $N=10,000$, the size of the cache to be $k=500$, and the queries' length to be $l=10$. Here, rather than letting the requested files follow the Zipf distribution or the random distribution, we generate a sequence of $100,000$ queries adversarially using the following rules:

\begin{itemize}
	\item Every query must request $9$ fixed files: $10000$, $9999$, $9998$, $9997$, $9996$, $9995$, $9994$, $9993$, $9992$, $9991$.
	\item For the remaining $1$ file in the query, we let the first query up to the $491$-st one request the file whose index matches their sequence number, i.e. the first query requests file $1$, the second query requests file $2$, and the $491$-st query requests file $491$.
	\item  Every $491$ queries form a cyclic sequence, that is, the $i$-th query requests exactly the same file as the $(i-491)$-th query, $\forall i\ge 492$.
\end{itemize}

Queries generated as such force the LRU algorithm to suffer a cache miss at every query, while the FF algorithm (which can be shown to be the exact optimal offline algorithm) suffers a cache miss only once every $491$ queries. Here, the goal is to show that although in many instances, the LRU algorithm performs well and even better than the query-wise marking algorithm, there are extreme instances where the LRU algorithm will perform very poorly. This indicates the existence of bad instances where the upper bound for the competitive ratio of the query-wise marking algorithm is much lower than that for the LRU algorithm. The result of this experiment is depicted in Figure \ref{fig5000}. It can be seen that the query-wise marking algorithm performs very good (near-optimal), while the LRU algorithm performs extremely poor.  
\begin{figure} 
	\subfigure[Transient phase]{ 
		\includegraphics[width=0.48\linewidth]{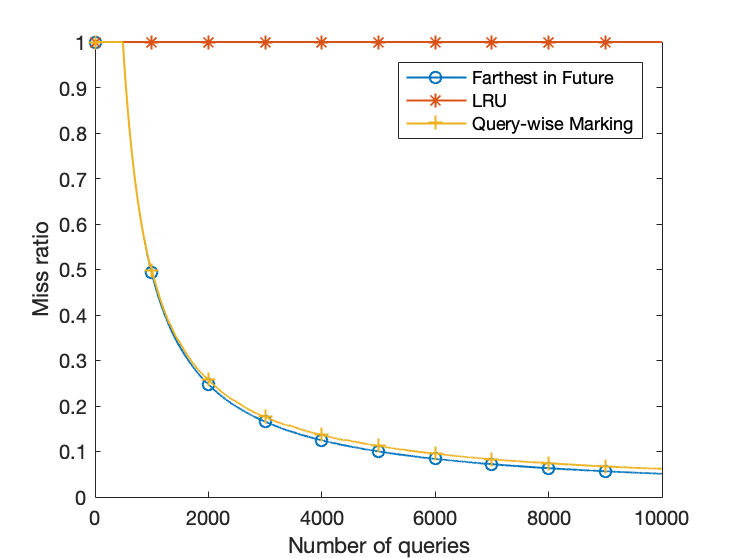}\label{fig5001}}
	\subfigure[Steady-state phase]{
		\includegraphics[width=0.48\linewidth]{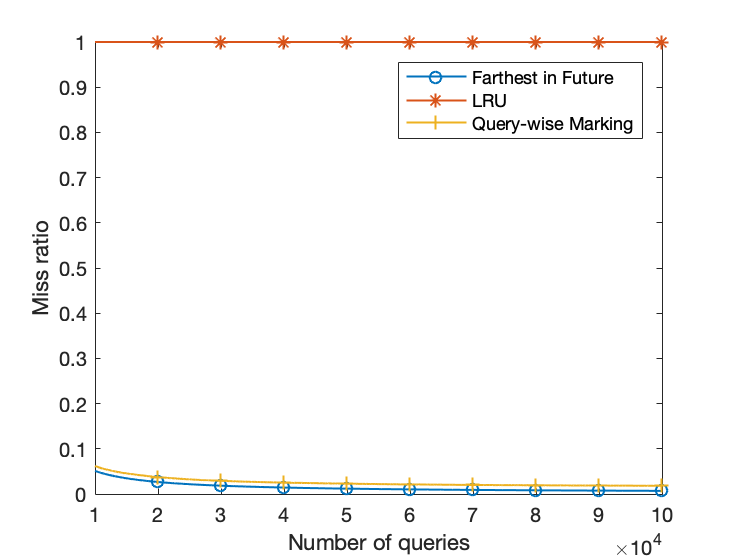}\label{fig5002}}
	\caption{Adversarially generated queries. Altogether the farthest in future algorithm incurred $694$ cache misses out of the $100,000$ queries, the LRU algorithm incurred $100,000$ cache misses, and the query-wise marking algorithm incurred $1823$ cache misses. Left: the transient phase that includes filling up the cache. Right: the steady-state phase. }\label{fig5000}
\end{figure}

\section{Conclusion}\label{sec:conclusion}

In this paper, we studied the file-bundle caching problem. We analyzed the performance of the well-known LRU/marking algorithms for the file-bundle caching problem. We showed that the LRU algorithm is $k$-competitive, and that is nearly the best among online deterministic algorithms. Moreover, we showed that the marking algorithm is $2l\cdot (\ln \frac{k}{k-h}-\ln \ln \frac{k}{k-h}+\frac{1}{2})$-competitive when $\frac{k}{k-h}\ge e$, and $2$-competitive otherwise, and these competitive ratios are within a factor of at most two from the optimal ones. Due to the nearly optimal performance of the LRU/marking algorithms in the case of single-cache file-bundle caching problem, we then used these algorithms to devise efficient deterministic/randomized distributed caching algorithms when there are multiple caches in the system. In particular, for $m=l+1$ caches, we developed a deterministic distributed caching algorithm which is $(l^2+l)$-competitive and a randomized distributed caching algorithm which is $2l\cdot(\ln(2l+1)-\ln\ln l+\frac{1}{2})$-competitive when $l\ge 2$. We also showed that the \emph{farthest in future} (FF) algorithm for the file-bundle caching problem achieves an approximation factor of at most $2l$ compared to the optimal offline algorithm, where $l$ is the length of the queries.

As a future direction of research, an interesting problem is to incorporate the role of different weights of the files into the problem. For instance, instead of setting a unit cost for any cache miss, we can set the cost of a cache miss to be the maximum weight of the files contained in the query (i.e., the $l_{\infty}$-norm of the difference vector). In fact, \cite{bansal2012primal} and \cite{buchbinder2019k} have given $O(\log(\frac{k}{k - h + 1}))$-competitive algorithms for the weighted $(h,k)$-paging problem. Therefore, one can extend these results to the generalized weighted $(h,k)$-paging problem under the file-bundle setting. Finally, we did not establish a lower bound for the competitive ratio of distributed caching algorithms in this work. It would be very interesting to establish tight lower bounds or to develop distributed caching algorithms with better competitive ratios.

\bibliographystyle{ACM-Reference-Format}
\bibliography{thesisrefs}

\section{Appendix I: Auxiliary Lemmas}

\begin{lemma}\label{lemm:max-T}
Let $T(m) := \frac{m+m\ln\frac{k}{m}}{\frac{k-h+m}{l}}$. Then $\max_mT(m)\leq l\cdot (\ln \frac{k}{k-h}-\ln \ln \frac{k}{k-h}+\frac{1}{2})$ if $\frac{k}{k-h}\ge e$, and $\max_mT(m)\leq l$, otherwise.
\end{lemma}
\begin{proof}
We have $T(m)=l\cdot\frac{m}{k-h+m}(1+\ln\frac{k}{m})$. Let $t := \frac{m}{k-h+m}$, and $s$ to be the constant $s: = \frac{k}{k-h}$. Then,
	\begin{align}\label{eq:T-G-Bound}
	T(m ) = G(t) &:= l\cdot(t+t\ln(s\cdot\frac{1-t}{t}))\cr
	&\le l\cdot(t_0+t_0\ln(s\cdot\frac{1-t_0}{t_0})),
	\end{align}
	where $t_0$ satisfies $\frac{d G}{d t}(t_o) = 0$, or equivalently
\begin{align}\label{eq:t0}
\frac{t_0}{1-t_0}+\ln\frac{t_0}{1-t_0} = \ln s.
\end{align}
Now if $ \frac{k}{k-h}\le e$, we have $0\le \ln s\le 1$ and $0\le t_0\le \frac{1}{2}$. Therefore, $T(m)\le G(t_0)\le l$. Otherwise, if $ \frac{k}{k-h}\ge e$, we have $s\ge e$ and $1\ge t_0\ge \frac{1}{2}$. Let us define $u := \frac{t_0}{1-t_0}$, which also implies $t_0=\frac{u}{1+u}$. Now we can rewrite \eqref{eq:T-G-Bound} and \eqref{eq:t0} in terms of the new variable $u$ as,
\begin{align}\nonumber
&G(t_0)\leq l(t_0+t_0\ln(s\frac{1-t_0}{t_0}))=l\frac{u}{1+u}(1-\ln u+\ln s),\cr 
&u+\ln u=\ln s.
\end{align}	
Combining these two relations we obtain $G(t_0)\leq lu$. Finally, for $s\ge e$ and since $u+\ln u=\ln s$, it is easy to see that $u\leq \ln s-\ln \ln s+\frac{1}{2}$. Thus,
\begin{align}\nonumber
\max_mT(m)= G(t_0)\leq lu\leq l(\ln s-\ln \ln s+\frac{1}{2}),
\end{align}
where we recall that $s = \frac{k}{k-h}$.
\end{proof}

\bigskip
\begin{lemma}\label{eq:lower-bound}
Let $f(m)= \frac{m}{k-h+l+1+m}\ln\frac{k+1}{k-h+l+1+m}$. Then 
\begin{align}\nonumber
\max_mf(m)\ge \ln \frac{k+1}{k-h+l+1}-\ln\ln\frac{k+1}{k-h+l+1}-2.
\end{align}
\smallskip
\end{lemma}
\begin{proof}
	Let us define $b:=k+1$, $c:=k-h+l+1$, and $u:=\frac{m}{c+m}$. We have
	$$f(m) = u\ln \frac{b}{c}(1-u):= g(u). $$
Let $m^*$ be the smallest integer $m^*\ge c\ln\frac{b}{c}-c$, such that $\frac{k-h+m^*}{l}$ is integral. Then, we have $c\ln\frac{b}{c}-c\le m^*\le c\ln\frac{b}{c}-c+l$, and consequently, $1-\frac{1}{\ln\frac{b}{c}}\le u^*\le 1-\frac{c}{c\ln\frac{b}{c}+l}$, where $u^*:=\frac{m^*}{c+m^*}$. This means that $u^*\ge 1-\frac{1}{\ln\frac{b}{c}}$ and $1-u^*\ge \frac{c}{c\ln\frac{b}{c}+l}\ge \frac{1}{e}\frac{1}{\ln\frac{b}{c}}$, and we can write,
	\begin{align*}
	g(u^*) &= u^*\ln \frac{b}{c}(1-u^*)\\
	&\ge (1-\frac{1}{\ln\frac{b}{c}})\ln\left(\frac{1}{e}\frac{b}{c\ln\frac{b}{c}}\right)\\
	&=-1+\frac{1}{\ln\frac{b}{c}}+\ln\frac{b}{c\ln\frac{b}{c}}-\frac{\ln\frac{b}{c\ln\frac{b}{c}}}{\ln\frac{b}{c}}\\
	&=\ln\frac{b}{c}-\ln\ln\frac{b}{c}-2+\frac{1}{\ln\frac{b}{c}}+\frac{\ln\ln\frac{b}{c}}{\ln\frac{b}{c}}\\
	&\ge \ln\frac{b}{c}-\ln\ln\frac{b}{c}-2.
	\end{align*}
Therefore, we have, 
	\begin{align*}
	&\max_mf(m)=\max_u g(u)\ge g(u^*)\ge \ln\frac{b}{c}-\ln\ln\frac{b}{c}-2\\
	&\qquad=l\ln\frac{k+1}{k-h+l+1}-l\ln\ln\frac{k+1}{k-h+l+1}-2l+1.
	\end{align*}
\end{proof}

\begin{theorem}\label{thrm6}
	The \emph{farthest in future} (FF) algorithm for the file-bundle caching problem achieves an approximation factor of at most $2l$ compared to the optimal offline algorithm, where $l$ is the queries' length.
\end{theorem}

\begin{proof}
	Again we consider the same phase partitioning of the request sequence $\sigma$:
	\begin{itemize}
		\item Phase 1 begins at the first page of $\sigma$;
		\item Phase $i$ begins at the first query which contains the $(k+1)$-th distinct page since phase $i-1$ has begun.
	\end{itemize}
	Consider any two consecutive phases $i$ and $i+1$, with phase $i+1$ including $m$ \emph{new pages}. This means there are also $m$ pages exclusive to phase $i$, and $k-m$ pages shared by both phase $i$ and phase $i+1$. First we want to prove that for thr FF algorithm, at the end of phase $i$ the cache contains at least $k-m$ of the $k$ distinct pages in phase $i+1$. In fact, we can prove that by analyzing two cases:
	\begin{itemize}
		\item[1] At least one page in the cache at the end of phase $i$ is not requested in both phase $i$ and phase $i+1$. In this case, since during phase $i$, that page should be evicted prior to all the pages that are requested in phase $i+1$, this means none of the $k-m$ common pages requested in phase $i$ is evicted. Therefore the cache contains at least $k-m$ of the $k$ distinct pages in phase$i+1$.
		\item[2] All the pages in the cache at the end of phase $i$ are requested in phase $i$ or phase $i+1$. Since there are only $m$ pages exclusive to phase $i$,  therefore the cache contains at least $k-m$ of the $k$ distinct pages in phase $i+1$.
	\end{itemize}
	Now, consider the FF algorithm and let $d(j)$ denote the number of pages that are not contained in the current cache but will be requested later in phase $i+1$, after the $j$-th query in phase $i+1$ has arrived. Since during phase $i+1$, the FF algorithm will not evict any pages that are contained in the current cache and will be requested later in phase $i+1$, it is easy to see that $d(j)$ is non-increasing, and $d(j)$ decreases by at least $1$ whenever a cache miss occurs. Notice that from the above discussion $d(0)\le m$, and at the end of phase $i+1$ we have $d(j)=0$. Thus, the number of cache misses of the FF algorithm during phase $i+1$ is at most $m$.	
	Finally, following the arguments in the proof of Theorem \ref{thrm1}, it is easy to see that the optimal algorithm makes at least $\frac{k-k+m}{l}=\frac{m}{l}$ cache missed during phases $i$ and $i+1$. This indicates that the competitive ratio of the FF algorithm is at most $2l$.
\end{proof}

\end{document}